\documentclass[11pt]{article}

\usepackage{amssymb}
\usepackage{amsmath}
\usepackage{amsthm}
\usepackage[lined,algoruled,shortend,linesnumbered]{algorithm2e}
\SetAlFnt{\small}
\SetAlCapFnt{\normalsize}
\SetAlCapNameFnt{\normalsize}
\usepackage{setspace}
\usepackage{xfrac}
\usepackage{graphicx}
\usepackage{epstopdf}
\usepackage{framed}
\usepackage{subcaption}
\usepackage[nohead, margin=0in]{geometry}
\usepackage{fullpage}
\newtheorem{thm}{Theorem}
\newtheorem{prop}[thm]{Proposition}

\newtheorem{cor}[thm]{Corollary}
\newtheorem{lem}[thm]{Lemma}

\makeatletter
\newcommand{\pushright}[1]{\ifmeasuring@#1\else\omit\hfill$\displaystyle#1$\fi\ignorespaces}
\newcommand{\pushleft}[1]{\ifmeasuring@#1\else\omit$\displaystyle#1$\hfill\fi\ignorespaces}
\makeatother

\title{Approximation of Steiner Forest via the Bidirected Cut Relaxation}

\author{Ali {\c{C}}ivril}

\begin{document}

\maketitle

\begin{abstract}
The classical algorithm of Agrawal, Klein and Ravi [SIAM J. Comput., 24 (1995), pp. 440-456], stated in the setting of the primal-dual schema by Goemans and Williamson [SIAM J. Comput., 24 (1995), pp. 296-317] uses the undirected cut relaxation for the Steiner forest problem. Its approximation ratio is $2-\frac{1}{k}$, where $k$ is the number of terminal pairs. A variant of this algorithm more recently proposed by K\"onemann et al. [SIAM J. Comput., 37 (2008), pp. 1319-1341] is based on the lifted cut relaxation. In this paper, we continue this line of work and consider the bidirected cut relaxation for the Steiner forest problem, which lends itself to a novel algorithmic idea yielding the same approximation ratio as the classical algorithm. In doing so, we introduce an extension of the primal-dual schema in which we run two different phases to satisfy connectivity requirements in both directions. This reveals more about the combinatorial structure of the problem. In particular, there are examples on which the classical algorithm fails to give a good approximation, but the new algorithm finds a near-optimal solution.
\end{abstract}

\section{Introduction}
The \textsf{Steiner forest} problem is one of the central problems in the field of approximation algorithms and network design. It is a natural generalization of the famous \textsf{Steiner tree} problem, and stands as the starting point for many other generalizations occupying a large fraction of network design literature. In this problem, one is given an undirected graph $G=(V,E)$, a cost function on the edges $c:E \rightarrow \mathbb{Q}^+$ and a set of terminal pairs $R=\{(s_1,t_1), \hdots, (s_k,t_k)\}$ (we set $n:=|V|$ and $m:=|E|$ throughout the paper). The objective is to find a subgraph $F$ of $G$ (which is necessarily a forest) of minimum cost $c(F) := \sum_{e \in F} c(e)$, which connects every terminal pair.

The \textsf{Steiner forest} problem had a pivotal role in the development of the fundamental techniques for the field of approximation algorithms, being the main problem of interest for the primal-dual schema with the idea of growing dual variables synchronously, which was introduced at the beginning of 90s. In particular, the famous approximation algorithm given by \cite{AKR} (henceforth called \texttt{AKR}), which has an approximation ratio of $2-\frac{1}{k}$ stimulated a series of results for similar problems, and in general for the area of network design and connectivity problems. This algorithm stated in purely combinatorial terms was then underlined with the language of the primal-dual schema by \cite{GW}, who introduced a more general approach for approximating such problems. Both of these approaches make use of the \emph{undirected cut relaxation} (UCR) for the problem, which has an integrality gap of at least $2-\frac{1}{k}$.

Due to its importance in the field of approximation algorithms, the problem of finding an algorithm for \textsf{Steiner forest} with a constant approximation ratio better than $2$ was stated as one of the top ten open problems in the area in a recent textbook by \cite{Williamson-Shmoys}. However, since the appearance of the conference paper by \cite{AKR-conf}, there have been no improved approximation algorithms discovered. Given that how much we know about its special case, the \textsf{Steiner tree} problem, for which there are many different LP relaxations and algorithmic techniques, it is of great interest to see if there are variations in algorithmic ideas for \textsf{Steiner forest} even if they do not provide significant improvements in terms of the approximation ratio. In this respect, the fact that there had been a single constant factor approximation algorithm for the problem for a long time is also intriguing. Along these lines, a recent attempt by \cite{Gupta-Kumar} proves a constant factor approximation for a greedy algorithm, a result which does not make use of an LP relaxation. A work similar in vein to this result followed by \cite{Gross} using local search.

This paper is motivated by the question of whether there are new LP relaxations for the \textsf{Steiner forest} problem, yielding novel algorithmic ideas. More relevant to this question, \cite{KLS} introduced a new LP relaxation called the \emph{lifted cut relaxation (LCR)}, motivated by a game-theoretic version of the problem. They show that LCR is stronger than UCR, and its integrality gap is at least $2-\frac{2}{k+1}$. The algorithm they present (henceforth called \texttt{KLS}), which computes a feasible dual with respect to LCR is a variant of \texttt{AKR} with a modified set of duals to be grown. Its approximation ratio is also $2-\frac{1}{k}$, although as the authors point out, the solution it returns is usually costlier than that of \texttt{AKR}.

The importance of delving more into the combinatorial structure of \textsf{Steiner forest} is also related to the more general \textsf{survivable network design} problem. Extensions of the usual approach inspired by \texttt{AKR} has only had limited success so far (\cite{GoemansGPSTW}; \cite{Williamson95}) toward the goal of a 2-approximation primal-dual algorithm for this problem.

\subsection{The results}
We introduce a new LP relaxation for the \textsf{Steiner forest} problem, which we call the \emph{bidirected cut relaxation} (BCR). This is inspired by the bidirected cut relaxation for the \textsf{Steiner tree} problem in which one replaces each edge by two arcs in both directions. It is an easy result that BCR is equivalent to UCR. We would like to stress the fact that our bidirected cut relaxation is not the same as the one introduced for Steiner tree (\cite{GoemansM93}; \cite{ChopraR94}), which can also be extended to the Steiner forest problem. Indeed, this relaxation has not been well exploited for both of the problems. In contrast, what we consider can be seen as a bidirected version of the usual undirected cut relaxation, which we use to construct \emph{two} paths between pairs in both directions.

Using our bidirected cut relaxation, we provide a new primal-dual algorithm for the \textsf{Steiner forest} problem with approximation ratio $2-\frac{1}{k}$. The algorithm is a novel extension of the primal-dual schema consisting of two phases with synchronous dual growth, one starting from the terminals $s_i$, and the other starting from the terminals $t_i$. We combine the results of these two phases followed by a standard pruning phase and a final reduction phase on certain subgraphs problematic for BCR. The proof of the approximation ratio also turns out to be quite different than the usual practice for primal-dual type algorithms. In the usual approach, the duals collide, and the set of edges they cover is considered by looking at the degrees of the duals. In our case, a set of duals growing against each other in different directions is considered.

To underline the differences between \texttt{AKR} and the new algorithm, we provide an example on which \texttt{AKR} and \texttt{KLS} give an approximation ratio arbitrarily close to $2-\frac{1}{k}$, whereas the new algorithm finds a near-optimal solution. We also provide a tight example for the new algorithm on which the approximation ratio is arbitrarily close to $2-\frac{1}{k}$.

Our approach enlarges the (small) set of $2$-approximation algorithms for the \textsf{Steiner forest} problem, which might stimulate new insights on how one can break the barrier of $2$, especially in light of the tight examples we present. More generally, given a problem with a cut-based relaxation involving terminal pairs, one can consider the bidirected version instead of the usual undirected one, thus possibly having a two-phase algorithm similar to the one in this paper. How widely applicable this is and whether it would yield improved results or new insights on a given problem is a question of interest.

\subsection{Organization}
The rest of the paper is structured as follows. In Section 2, we review the undirected cut relaxation together with \texttt{AKR} exploiting it. Section 3 introduces the bidirected cut relaxation for \textsf{Steiner forest} and the new primal-dual algorithm for \textsf{Steiner forest} using this relaxation. Section 4 provides the details of a straightforward implementation. In Section 5, we establish the approximation ratio of the new algorithm. In Section 6, we give the aforementioned tight examples.
\section{The undirected cut relaxation and \texttt{AKR}}
\begin{algorithm}[!b]
\caption{\texttt{AKR}$(G=(V,E), R, c)$}
 $y \leftarrow 0$ \\
 $F \leftarrow \emptyset$ \\
 $\ell \leftarrow 0$
 \BlankLine \BlankLine
 // The augmentation phase \\
 \While{not all $s_i$-$t_i$ pairs are connected in $(V,F)$} {
 $\ell \leftarrow \ell+1$ \\
 Let $\mathcal{C}$ be the set of all connected components $C$ of $(V,F)$ such that $|C \cap \{s_i,t_i\}| = 1$ for some $i$ \\
 Increase $y_C$ for all $C \in \mathcal{C}$ uniformly until for some $e_{\ell} \in \delta(C'), C' \in \mathcal{C}, c(e_{\ell}) = \sum_{C:e_{\ell} \in \delta(C)} y_C$ \\
 $F \leftarrow F \cup \{e_{\ell}\}$
 }
 \BlankLine \BlankLine
 // The pruning phase \\
 $F' \leftarrow F$ \\
 \For{$j \leftarrow \ell$ downto $1$} {
 \If{$F'-\{e_j\}$ is feasible} {
 $F'\leftarrow F'-\{e_j\}$
 }
 }
 \BlankLine \BlankLine
 \Return $(F',y)$
\end{algorithm}

Let $\mathcal{S}$ be the set of subsets $S$ of $V$ that separate at least one terminal pair in $R$. In other words, $S \in \mathcal{S}$ if and only if there is $(s,t) \in R$ satisfying $|S \cap \{s,t\}| = 1$. We call an element in $\mathcal{S}$ a \emph{Steiner cut} or simply a \emph{cut}. Let $\delta(S)$ denote the set of edges with exactly one endpoint in $S$. The undirected cut relaxation of the problem is then as follows:
\begin{alignat*}{4}
\text{minimize}   \qquad & \sum_{\substack{e \in E}} c(e) x_e & & \pushright{\text{(UCR)}} \\
\text{subject to} \qquad & \sum_{\substack{e \in \delta(S)}} x_e \geq 1, \qquad &\forall S \in \mathcal{S}, \\
            &  x_e \geq 0, &\forall e \in E.
\end{alignat*}

\noindent The dual of this linear program is
\begin{alignat*}{4}
\text{maximize}   \qquad & \sum_{\substack{S \in \mathcal{S}}} y_S & & \pushright{\text{(UCR-D)}} \\
\text{subject to} \qquad & \sum_{\substack{S \in \mathcal{S}}: e \in \delta(S)} y_S \leq c(e), \qquad &\forall e \in E, \\
            & y_S \geq 0, &\forall S \in \mathcal{S}.
\end{alignat*}

\texttt{AKR} synchronously grows dual variables corresponding to the cuts separating any pair. The sets corresponding to these cuts, which are selected to be minimal with respect to inclusion are called \emph{minimal violated sets}. It iteratively improves the feasibility of the primal solution by taking edges whenever the corresponding constraints become tight. After arriving at a primal feasible solution, it removes the unnecessary edges, i.e. the edges whose removal do not violate the feasibility, in the reverse order of their inclusion. The following is a standard result from \cite{GW}:
\begin{thm}[\cite{GW}]
\label{gw}
If $F'$ and $y$ are the set of edges and the dual variables returned by \texttt{AKR}, then
$$
\sum_{e \in F'} c(e) \leq \left(2-\frac{2}{|A|}\right) \cdot \sum_{S \subseteq V} y_S \leq \left(2-\frac{1}{k}\right) \cdot \sum_{S \subseteq V} y_S,
$$

\noindent where $A$ is maximum number of minimal violated sets during the algorithm.
\end{thm}

\section{The bidirected cut relaxation and the new primal-dual algorithm}
We first replace each edge $e=\{u,v\} \in V$ by two directed arcs $(u,v)$ and $(v,u)$ each with cost $\frac{1}{2}c(e)$. For a given cut $S \subseteq V$, we define $\delta^+(S) = \{(u,v) \in E| u \in S, v \notin S\}$, i.e. the set of arcs emanating from $S$. As usual, we set $\mathcal{S}$ be the set of cuts that separate at least one terminal pair: $S \in \mathcal{S}$ if and only if there is $(s,t) \in R$ satisfying $|S \cap \{s,t\}|=1$. Then, the following is a relaxation for the \textsf{Steiner forest} problem.
\begin{alignat*}{4}
\text{minimize}   \qquad & \frac{1}{2} \sum_{e \in E} c(e) x_e & & \pushright{\text{(BCR)}} \\
\text{subject to} \qquad & \sum_{e \in \delta^+(S)} x_e \geq 1, \qquad
&\forall S \in \mathcal{S}, \\
            & x_e \geq 0, &\forall e \in E.
\end{alignat*}

\noindent It is a straightforward result that (BCR) is equivalent to (UCR), i.e. they can be converted to each other with equal objective values by assigning appropriate values to the edges/arcs. In particular, for converting (BCR) to (UCR), the value of an undirected edge is assigned to the corresponding directed arcs. From (BCR) to (UCR), we assign the average of the values of the arcs in both directions to the corresponding undirected edge. Thus, the integrality gap of (BCR) is also $2-\frac{1}{k}$.

Let us now write the dual of the linear program (BCR).
\begin{alignat*}{4}
\text{maximize}   \qquad & \sum_{S \in \mathcal{S}} y_S & & \pushright{\text{(BCR-D)}} \\
\text{subject to} \qquad & \sum_{\substack{S: e \in \delta^+(S)}} y_S \leq \frac{1}{2} c(e), \qquad &\forall e \in E, \\
            & y_S \geq 0, &\forall S \in \mathcal{S}.
\end{alignat*}

Similar to \texttt{AKR}, the new algorithm is also based on the primal-dual schema and growing dual variables in a synchronized fashion. However, since the underlying graph is a bidirected graph, the algorithm tries to satisfy the constraints of the primal program (BCR) by constructing a solution in both directions. This requires two distinct phases for the selection of arcs. In total, the algorithm consists of four phases to produce a feasible solution. In the first phase, we grow the dual variables starting from the terminals $s_i$, and continue the usual process of including arcs that go tight until there are directed paths from each $s_i$ to $t_i$. Note that this does not necessarily make a feasible solution as some arcs might only be taken in one direction. The solution constructed in the first phase is an input to the second phase in which we apply the same procedure, but this time starting to grow the dual variables from the terminals $t_i$. We continue until there are directed paths from each $t_i$ to $s_i$. By definition, the solution constructed in the first two phases contains bidirected paths between each terminal pair.
\begin{algorithm}[!]
\caption{\textsc{Bidirected-Primal-Dual($G=(V,E)$, $R$, $c$)}}
 // Initialization \\
 $y \leftarrow 0$ \\
 $F \leftarrow \emptyset$ \\
 $\ell \leftarrow 0$
 \BlankLine \BlankLine
 // The first augmentation phase \\
 \While{there are terminal pairs in $R$ not connected by a directed $s_i$-$t_i$ path in $(V,F)$} {
 $\ell \leftarrow \ell+1$ \\
 Let $\mathcal{C}$ be the set of all minimal sets $C$ (w.r.t. inclusion) such that $|\delta^+(C) \cap F|=0$, and $s_i \in C$, but $t_i \notin C$ for some $i$ \\
 Increase $y_C$ for all $C \in \mathcal{C}$ uniformly until for some $e_{\ell} \in \delta^+(C')$, $C' \in \mathcal{C}$, $c(e_{\ell}) = \sum_{S:e_{\ell} \in \delta^+(C)} y_C$ \\
 $F \leftarrow F \cup \{e_{\ell}\}$
 }
 \BlankLine \BlankLine
 // The second augmentation phase \\
 \While{there are terminal pairs in $R$ not connected by a directed $t_i$-$s_i$ path in $(V,F)$} {
 $\ell \leftarrow \ell+1$ \\
 Let $\mathcal{C}$ be the set of all minimal sets $C$ (w.r.t. inclusion) such that $|\delta^+(C) \cap F|=0$, and $t_i \in C$, but $s_i \notin C$ for some $i$ \\
 Increase $y_C$ for all $C \in \mathcal{C}$ uniformly until for some $e_{\ell} \in \delta^+(C')$, $C' \in \mathcal{C}$, $c(e_{\ell}) = \sum_{S:e_{\ell} \in \delta^+(C)} y_C$ \\
 $F \leftarrow F \cup \{e_{\ell}\}$
 }
 \BlankLine \BlankLine
 // The pruning phase \\
 $F' \leftarrow F$ \\
 \For{$j \leftarrow \ell$ downto $1$} {
 \If{$F'-\{e_j\}$ is feasible} {
 $F' \leftarrow F'-\{e_j\}$
 }
 }
 $F^1 \leftarrow \{(u,v) \in F'| (v,u) \in F'\}$ \\
 $F' \leftarrow F'-F^1$
 \BlankLine \BlankLine

 // The reduction phase \\
 $F^2 \leftarrow \emptyset$ \\
 Let $\{(s_i',t_i')\}$ be the set of pairs such that there are disjoint bidirected paths between $s_i'$ and $t_i'$ in $F'$ AND at least one of the following holds for $v \in \{s_i',t_i'\}$: \\
 (1) $v$ is adjacent to some edge in $F^1$ \\
 (2) $v \in \{s_i,t_i\}$ for some $i \in \{1,\hdots,k\}$ \\
 \For{all pairs $(s_i',t_i')$} {
 Let $P_s$ be the directed path $s_i'-t_i'$ \\
 Let $P_t$ be the directed path $t_i'-s_i'$ \\
 $P \leftarrow \arg\min_{P \in \{P_s, P_t\}} \tau(P)$ \\
 Double the arcs in $P$ by adding the ones in reverse direction \\
 $F^2 \leftarrow F^2 \cup P$
 }
 \BlankLine \BlankLine
 $F^3 \leftarrow F^1 \cup F^2$ \\
 \Return $(F^3, y)$
\end{algorithm}

As in the case of \texttt{AKR}, the set of dual variables that are grown at a particular phase in the new algorithm must naturally satisfy certain properties. Given a cut $S$ (synonymously a \emph{dual}) determined by a set of vertices, an already selected set of edges $F$, we say that $S$ is a \emph{minimal violated set} if there is at least one $(s,t) \in R$ such that $|S \cap \{s,t\}|=1$, $|\delta^+(S) \cap F|=0$, and $S$ is minimal with respect to inclusion. Note that the implications of these conditions are different from those of \texttt{AKR} based on the undirected cut relaxation. In that case, one can simply take the connected components $S$ of $(V,F)$ satisfying the property that $|S \cap \{s,t\}|=1$ for some $(s,t)$ which correspond to minimal violated sets. However, determining the minimal violated sets is not easy in our case since the underlying graph is directed. In particular, the minimal violated sets in the new algorithm are not necessarily disjoint (See the leftmost picture in Figure~\ref{merging}). We also make the distinction between the two different phases of the algorithm and say that a set $S$ satisfying the usual conditions stated above is an \emph{$s$-minimal violated set} if in addition it contains at least one $s_i$ but not $t_i$ for some valid $i$. In this case, we say that the corresponding dual \emph{originates from $s_i$}. Similarly, we say that it is a \emph{$t$-minimal violated set} if it contains at least one $t_i$ but not $s_i$, and we say that the corresponding dual originates from $t_i$. With this terminology, we are interested in growing the $s$-minimal violated sets in the first phase, and the $t$-minimal violated sets in the second phase.

The third phase which we call the pruning phase considers the arcs in the reverse order of their inclusion and discards an arc unless its exclusion violates the feasibility. The order is determined by the inclusion of the arcs in the first augmentation phase followed by the second augmentation phase. The arcs selected in both directions and which remain after this phase make the set $F^1$. These are in the final solution. At the end of the phase, we have the set of arcs $F'$ which are only selected in one direction.

The effect of the pruning phase of the new algorithm is quite different than that of \texttt{AKR}. In our case, it may not be clear which edges of the original input graph we should select even though the result is feasible. An example is given in Figure~\ref{bad} with the set of arcs shown in $F^1 \cup F'$, the feasible solution by the end of the pruning phase. The nontrivial duals grown are also shown in dashed lines. To see that we might have such an instance, note that there are two dual variables running on the arc with cost $1/2+\epsilon$, one on the $t_1$-side and the other on the $t_2$-side. This results in the inclusion of that arc before the arcs of cost $1$ are taken. In contrast, the arcs of cost $1$ are included in the first phase before the relevant dual covers the arcs of cost $1/2+\epsilon$ and $1/2$. The last phase is run as a final remedy for this situation.
\begin{figure}[!t]
\begin{center}
\includegraphics[width=.55\textwidth]{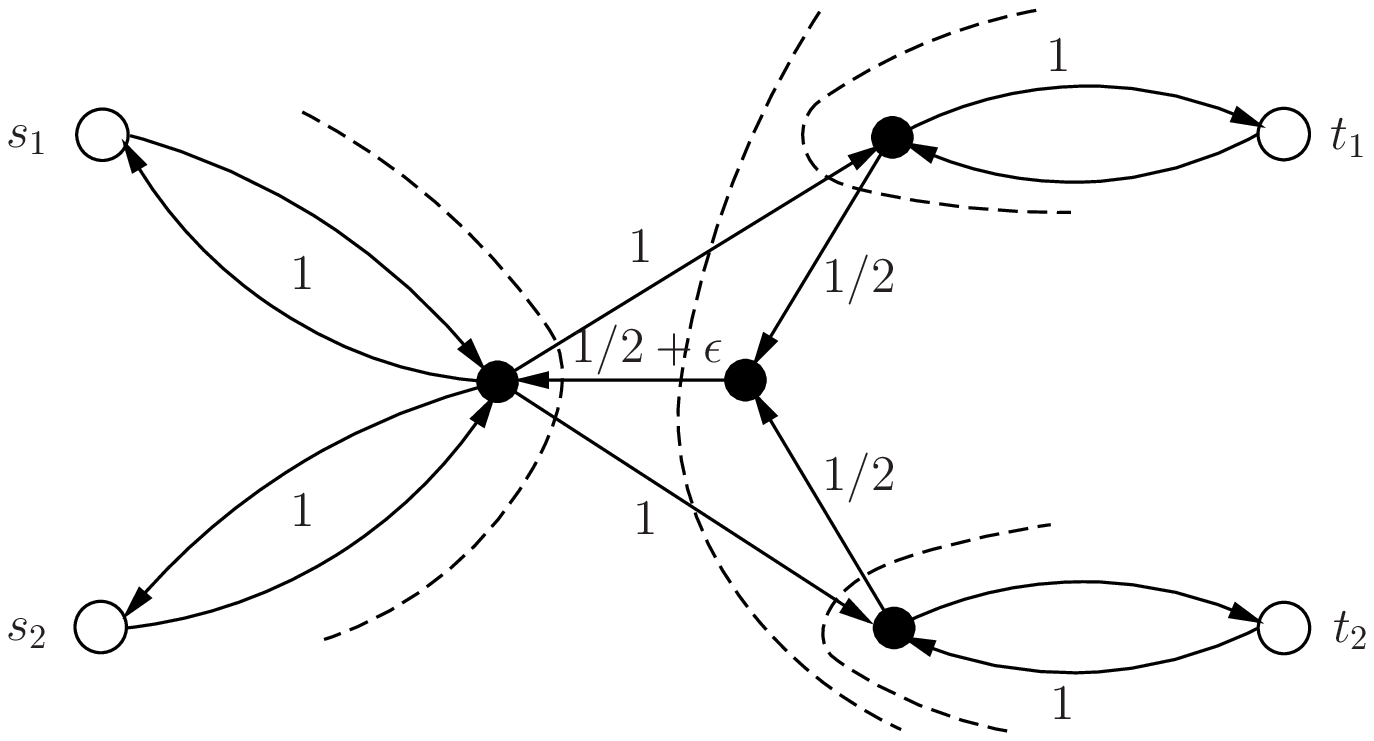}
\end{center}
\caption{An example on which the pruning phase does not lead to a valid solution}
\label{bad}
\end{figure}

In order to state the fourth phase, we need to consider a specific meaning for the growth of duals in the algorithm. A useful intuition for the type of primal-dual approach we utilize is to consider the growth of the duals in each iteration as a continuous process over time. Considering an arc along which a dual grows as a line segment, within a period of time $\epsilon > 0$, the dual is considered to ``cover'' a cost of $\epsilon$ of the arc, starting from the already covered part and continuously extending. An arc of unit length might go tight after half a unit of time if there are two duals growing along it. Accordingly, given any $\epsilon > 0$, if there are $d$ duals growing along an arc $e$ in the period of time $\epsilon$, we think that the partial cost of $\epsilon d \leq \frac{1}{2} c(e)$ is covered by the duals in that period. Given two vertices $u$ and $v$, and the directed path $P = u-v$ between them, we denote the period of time from the moment a dual is formed including $u$ to the moment a dual is formed including $v$ by $\tau(P)$.

The fourth phase, which we call the reduction phase determines the set of edges to be selected based on the information in $F'$. We consider all the pairs $(s_i',t_i')$ such that there are \emph{node disjoint} bidirected paths between $s_i'$ and $t_i'$ in $F'$, together with the requirement that at least one of the following holds for $v \in \{s_i',t_i'\}$.
\begin{itemize}
\item $v$ is adjacent to some edge in $F^1$ (i.e. the edges induced by the arcs taken in both directions);
\item $v \in \{s_i,t_i\}$ for some $i$ (i.e. the original set of pairs).
\end{itemize}

\noindent These are precisely the endpoints of the subgraphs that we seek a valid solution on. The algorithm considers both of the directed paths between such pairs and takes the path with a smaller $\tau$ value, i.e. the path which goes tight in a shorter period of time. It doubles the selected arcs by taking them in both directions and includes them into the solution $F^2$, which is feasible by definition. The final solution is the union of $F^1$ and $F^2$.

To give an example of the last phase, we consider again the graph given in Figure~\ref{bad}. $F'$ at the end of the pruning phase consists of the arcs forming the disjoint directed paths between the intermediate vertices. Note that there are two duals of the second phase growing along the arc of cost $1/2+\epsilon$. The time to cover the arcs directed from the $t$-side to the $s$-side is then $1/2+\frac{1}{2}(1/2+\epsilon) = 3/4+\epsilon/2$. The time to cover the arcs of cost $1$ directed from the $s$-side to the $t$-side on the other hand is $1$ since there is a single dual growing in the first phase. So the arcs of cost $1/2$ and $1/2+\epsilon$ are taken by the reduction phase.
\section{Implementation details}
We will give in this section a straightforward implementation of the algorithm. There may be faster and more compact implementations, which we leave as an open problem.

During the course of the algorithm, we explicitly store all the nodes in a given minimal violated set. Initially in both the first phase and the second phase, there are $k$ such lists and each list contains a single terminal representing a minimal violated set. By the execution of the algorithm, this number is non-increasing. Consequently, we have at most $k$ minimal violated sets at any time in the algorithm. We describe how to select the next arc and how to update the minimal violated sets for the first phase of the algorithm. The running times will be identical in the second phase.

In order to find the next tight arc, we keep a priority queue for arcs. The key values of the arcs are the times at which they will go tight. Initially, all the arcs that are not incident to the terminals might be set to $\infty$, and the key values of the immediately accessible arcs are set to their correct values examining their costs. For each arc, we also keep a list of duals growing on that arc. This is convenient in updating the key values. The initialization of the priority queue takes $O(m)$ time. At each iteration of the loop, we extract the minimum from the priority queue and update all the other arcs in the queue with the information obtained from the new set of minimal violated sets. This takes at most $O(m\log n)$ time since we consider at most $m$ arcs to update. In practice, this number might be much smaller.

Updating the minimal violated sets is the most expensive part of the algorithm. Upon inclusion of the arc in the current iteration, we update the list of nodes in the sets by performing a standard graph traversal procedure such as BFS, which takes time $O((m+n)k)=O(mk)$. Notice that not all of these sets might be minimally violated, i.e. there might be a set which is a proper subset of another. Initially declare all the sets \emph{active}, i.e. consider them as minimal violated sets.  In order to determine which one of these are actual minimal violated sets, we perform the following operation starting from the smallest cardinality set (assume that the lists keep their sizes). Compare the elements in the set with all the other sets, and if another set turns out to be a strict superset of this set, declare the larger set \emph{inactive}, i.e. not a minimal violated set. Comparing sets can be performed in expected time $O(n)$ by hashing the values of one set and looping over the second set to see if they contain the same elements. Hence, for a single set, we spend $O(nk)$ time in expectation. The total time requirement for this operation is then $O(nk^2)$. If the two sets compared are identical, we \emph{merge} them into a new minimal violated set and declare it active (See Figure~\ref{merging} for an example of this procedure and merging). The number of iterations of the main loop of the algorithm is at most $O(n)$. So the execution of the whole loop takes time $O(mn\log{n}+mnk+n^2k^2)$ in expectation.

In order to implement the pruning phase, we iterate over the arcs in $F'$.  For each such arc, we check for each terminal pair if they are still connected even if the arc is discarded. This takes time $O((m+n)k=O(mk)$ with a standard graph traversal algorithm. Since there are at most $O(n)$ arcs to consider, the total running time is then $O(mnk)$. The reduction phase amounts to re-executing the algorithm on a subgraph with an extra bookkeeping of times. Its running time is absorbed by
that of the whole algorithm. Thus, the algorithm can overall be implemented in time $O(mn\log{n}+mnk+n^2k^2)$.
\begin{figure}[!t]
\begin{center}
\begin{subfigure}{.32\textwidth}
  \centering
  \includegraphics[width=\linewidth]{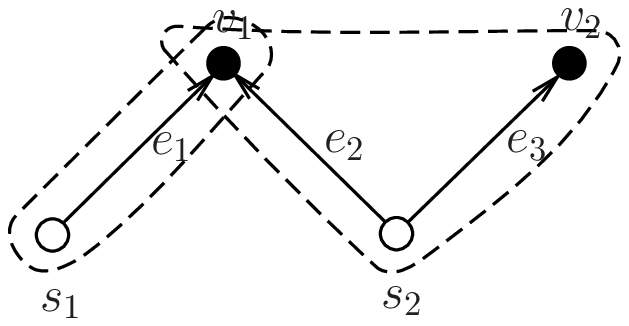}
  \label{merging-1}
\end{subfigure}
\begin{subfigure}{.32\textwidth}
  \centering
  \includegraphics[width=\linewidth]{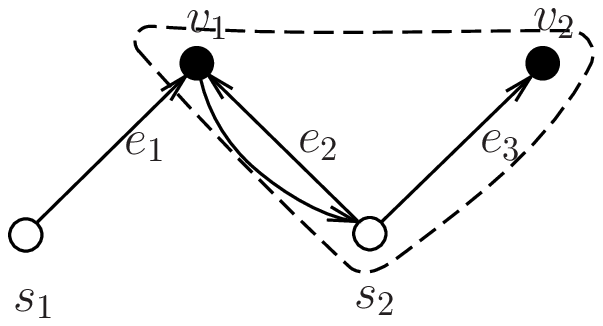}
  \label{merging-2}
\end{subfigure}
\begin{subfigure}{.32\textwidth}
  \centering
  \includegraphics[width=\linewidth]{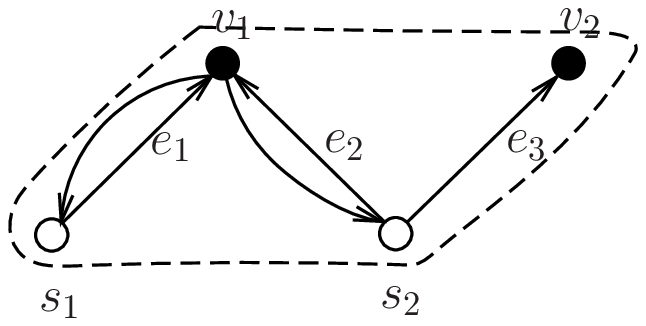}
  \label{merging-3}
\end{subfigure}
\end{center}
\vspace*{-5mm}
\caption{An example of a merging}
\label{merging}
\end{figure}

The merging of minimal violated sets is more difficult in the new algorithm compared to \texttt{AKR} since they do not necessarily merge even if they reach a common vertex. However, merging of minimal violated sets can occur during the algorithm, particularly when the algorithm takes some arcs in both directions. This is illustrated in Figure~\ref{merging}. The dual growing from $s_1$ takes the arc $e_1$ in forward direction, which we denote by $e_1^+$. The list of nodes representing this dual then becomes $S_1=\{s_1,v_1\}$. The dual growing from $s_2$ takes the arcs $e_2^+$ and $e_3^+$, making its node list $S_2=\{s_2,v_1,v_2\}$. When $S_1$ continues to grow to take $e_2^-$, its node list is updated to $S_1=\{s_1,v_1,s_2,v_2\}$, the set of nodes reachable from $s_1$. However, this cannot be a minimal violated set as $S_2$ is a proper subset of this list. As a result, the only minimal violated set in this iteration is $S_2=\{s_2,v_1,v_2\}$. After some time, $S_2$ takes $e_1^-$ and adds $s_1$ to its node list. At this time, the algorithm realizes that $S_1$ and $S_2$ are the same, i.e. $\{s_1,s_2,v_1,v_2\}$, and merges them to a new minimal violated set.
\section{Proof of the approximation ratio}
Recall that, given \emph{any time}, one can see the set of duals and their positions as a snapshot of the algorithm. In the following discussion proving Propositions~\ref{against-1}-\ref{against-3} and Lemmas~\ref{dual-one}-\ref{tree-2}, we refer the behavior of duals in an infinitesimal amount of time in which the snapshot remains the same.

With an abuse of notation, we denote by $F^3$ the set of undirected edges induced by the solution when we consider it as a subset of the original input graph. Given an iteration of the first phase and a dual $C^s$, we will consider the set of edges $\Delta(C^s) \cap F^3$, where $\Delta(C^s)$ denote the undirected set of edges induced by $\delta^+(C^s)$. We make the same definitions for a dual $C^t$ grown in an iteration of the second phase. Throughout this section, we say that duals grow along \emph{edges}, rather than arcs. This is for simplicity of discussion as we usually consider duals growing along the two arcs representing the same edge.

Given an undirected edge $e=\{u,v\} \in F^3$, we consider $e$ as a line segment defining an interval $[u=0,v=\frac{1}{2}c(e)]$. A single dual $C^t$ growing on $e^-=(v,u)$ (from $v$ to $u$) is considered to be \emph{grown against} a single dual $C^s$ growing on $e^+=(u,v)$ (from $u$ to $v$) if there is an interval $[a,b] \subseteq [u,v]$ such that both $C^s$ and $C^t$ grow on this interval. Given a dual $C^s$ grown in the first phase of the algorithm, the set of all duals grown against $C^s$ on $e$ for all $e \in \Delta(C^s) \cap F^3$ is called the set of \emph{duals grown against $C^s$}. We make similar definitions for a dual $C^t$ growing in an iteration of the second phase. Note that if $e \in F^1$, it is possible that a dual grown against $C^s$ belongs to the set of duals grown in the first phase. This happens, for instance, when two $s$-terminals are closer to each other than any other terminals. Similarly, a dual grown against $C^t$ might have grown in the second phase.

There may be multiple duals growing against each other on an edge. Accordingly, we make the following refinement over the definition above. Consider the case where a set of duals $\{C^{t_1}, \hdots, C^{t_{\beta}}\}$ is grown against $\{C^{s_1}, \hdots, C^{s_{\alpha}}\}$ on an edge $e$ within a period of time $\epsilon$. We define the \emph{graph of duals} on the edge $e$ as a single edge between two vertices. One of the vertices $C^s$ corresponds to the set $\{C^{s_1}, \hdots, C^{s_{\alpha}}\}$, and is assigned a \emph{growth speed} of $\alpha$. The other vertex $C^t$ represents the set $\{C^{t_1}, \hdots, C^{t_{\beta}}\}$ with a growth speed of $\beta$. We will later generalize the notion of graph of duals.

The proof of the performance ratio relies on the properties of the set of duals grown against each other. First, we make sure that there always exists a dual grown against another one.
\begin{prop}
\label{against-1}
(a) If $e^+=(u,v) \in \delta^+(C^s) \cap F^1$ for some $C^s$ grown in the first phase, then there is a dual grown against $C^s$ on $e=\{u,v\}$.

(b) If $e^-=(v,u) \in \delta^+(C^t) \cap F^1$ for some $C^t$ grown in the second phase, then there is a dual grown against $C^t$ on $e=\{u,v\}$.
\end{prop}

\begin{proof}
The first part of the statement holds by the definition of the algorithm since $e^-=(v,u) \in F^1$ is included into the solution either in the first phase or in the second phase. The second part is symmetric to the first one.
\end{proof}

\noindent Proposition~\ref{against-1} does not hold for the edges in $F^2$ since the edges in the source graph $F'$ are not selected in both directions. We thus properly define a new set of duals growing against each other on $F^2$, first considering a single edge. Let $P^+ = s'-t'$ and $P^- = t'-s'$ be the disjoint paths between a pair $(s',t')$ considered in the reduction phase of the algorithm. Assume without loss of generality that $\tau(P^+) \leq \tau(P^-)$, i.e. $P^+$ with doubled edges is selected by the algorithm. Let $e^+=(u,v) \in P^+$ be any edge selected via the duals grown in the first phase. Taking $e$ as a line segment $[u,v]$, let the set of duals grown along some $[a,b] \subseteq [u,v]$ be $\{C_1^s, \hdots, C_{\alpha}^s\}$ such that they originate from $s_1', \hdots, s_{\alpha}'$, respectively, and $e^+$ is taken by the reduction phase between $s_j'$ and $t_j'$ for all $1 \leq j \leq \alpha$. Note that there exists at least one such $j$ by definition. We create a set of $\alpha$ new duals, namely $\{C_1^t, \hdots, C_{\alpha}^t\}$ each growing against $C_j^s$ along the interval $[a,b]$.

Consider now defining a new set of duals described as above for all pairs $(s_i',t_i')$ connected by a path $P_i^+$ on which we select an edge $e_i^+$ and take an interval covered within a period of time $\epsilon$. We consider all these duals $y_S'$ for $F^2$ in the rest of our discussion, unless we explicitly mention the duals $y_S$ computed by the algorithm. So, considering $y_S'$, we have
\begin{prop}
\label{against-2}
(a) If $e^+=(u,v) \in \delta^+(C^s) \cap F^2$ for some $C^s$ grown in the first phase, then there is a dual grown against $C^s$ on $e=\{u,v\}$.

(b) If $e^-=(v,u) \in \delta^+(C^s) \cap F^2$ for some $C^t$ grown in the second phase, then there is a dual grown against $C^t$ on $e=\{u,v\}$.
\end{prop}

The edges in $F^2$ are now tight with respect to $y_S'$, i.e. their total cost is exactly covered by the duals in $y_S'$. Let $\alpha$ be the number of new duals defined on an interval covered within a period of time $\epsilon$, as described above. Each new dual covers a cost of $\epsilon$. This cost is compensated by one of the actual duals grown by the algorithm from some $t_j'$ to $s_j'$, particularly by covering the same portion of the cost of an edge on the path $t_j'-s_j'$. Indeed, since $e_i^+$ is taken by the algorithm on which $C_j^s$ grows, for the paths $P_j^+=s_j'-t_j'$ and $P_j^-=t_j'-s_j'$, we have $\tau(P_j^+) \leq \tau(P_j^-)$, and such a dual always exists. In particular, this implies
\begin{prop}
\label{against-3}
For every new set of $\alpha$ duals in $y_S'$ defined on an interval in $F^2$, each with value $\epsilon$, there is a distinct dual in $y_S$ of value $\epsilon$ computed by the algorithm.
\end{prop}

The following two lemmas use Proposition~\ref{against-1} and Proposition~\ref{against-2} as premises.
\begin{lem}
\label{dual-one}
Given a dual $C^s$ growing in an iteration of the first phase, let $\mathcal{C}^t$ be the set of duals grown against $C^s$. Then, for any $C^t \in \mathcal{C}^t$,
$$
|\Delta(C^s) \cap \Delta(C^t) \cap F^3| = 1.
$$
\end{lem}

\begin{proof}
If $|\Delta(C^s) \cap F^3| = 1$, there is nothing to prove. Thus, assume $|\Delta(C^s) \cap F^3| > 1$, and assume further for a contradiction that for some $C^t \in \mathcal{C}^t$, we have $|\Delta(C^s) \cap \Delta(C^t) \cap F^3| > 1$. Let $\{v_1,w_1\}$ and $\{v_2,w_2\}$ be two of the edges on which both $C^s$ and $C^t$ grow with $v_1,v_2 \in C^s$, $w_1,w_2 \in C^t$.

If $C^t$ has grown in the second phase, observe that there is an index $i$ such that $s_i \in C^s$ and $t_i \in C^t$. There is also an index $j$ such that $s_j \in C^s$ and $t_j \in C^t$. Otherwise, one of $\{v_1,w_1\}$ and $\{v_2,w_2\}$ would be redundant, resulting in a deletion in the pruning phase. Thus, we may assume without loss of generality that $\{v_1,w_1\}$ is on the path between $s_i$ and $t_i$, $\{v_2,w_2\}$ is on the path between $s_j$ and $t_j$. We further observe that there must be a bidirected path between $s_i$ and $s_j$ in $F_{\ell}$. Otherwise, it would contradict the minimality of $C^s$. Similarly, there is a bidirected path between $t_i$ and $t_j$. The existence of all these paths implies that there is a cycle in $F^3$, which is a contradiction (See Figure~\ref{dual-one-fig} for an illustration, where the cycle contains both of the terminal pairs).

If $C^t$ has grown in the first phase, then there are $s_i \in C^s$ and $s_{j_1},s_{j_2} \in C^t$. We may assume that $\{v_1,w_1\}$ is on the path between $s_i$ and $s_{j_1}$, $\{v_2,w_2\}$ is on the path between $s_i$ and $s_{j_2}$. By the minimality of $C^t$, there must also be a bidirected path between $s_{j_1}$ and $s_{j_2}$. This again induces a cycle, yielding a contradiction.
\end{proof}

\begin{figure}[!t]
\begin{center}
\includegraphics[width=.45\textwidth]{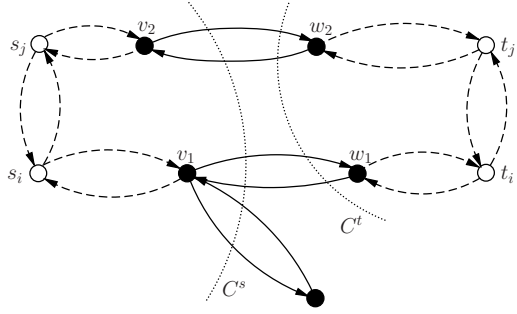}
\end{center}
\caption{An example illustrating the proof of Lemma~\ref{dual-one}}
\label{dual-one-fig}
\end{figure}

The symmetric result with a proof identical to that of Lemma~\ref{dual-one} except the interchanged roles of $C^s$ and $C^t$ is as follows.

\begin{lem}
\label{dual-one-t}
Given a dual $C^t$ growing in an iteration of the second phase, let $\mathcal{C}^s$ be the set of duals grown against $C^t$. Then, for any $C^s \in \mathcal{C}^s$,
$$
|\Delta(C^t) \cap \Delta(C^s) \cap F^3| = 1.
$$
\end{lem}

The following theorem establishes the approximation ratio of the algorithm by weak duality.
\begin{thm}
If $(F^3,y)$ is the solution returned by the new primal-dual algorithm for the \textsf{Steiner forest} problem, then
$$
\frac{1}{2} \sum_{e \in F^3} c(e) \leq \left(2-\frac{1}{k}\right) \cdot \sum_{S \subseteq V} y_S.
$$
\end{thm}

\begin{proof}
Since all the edges in $F^1$ are tight, we have
$$
\frac{1}{2} \sum_{e \in F^1} c(e) = \sum_{e \in F^1} \sum_{\substack{S:e \in \delta^+(S)}} y_S.
$$

\noindent The edges in $F^2$ are also tight with respect to $y_S'$, i.e.
$$
\frac{1}{2} \sum_{e \in F^2} c(e) = \sum_{e \in F^2} \sum_{\substack{S:e \in \delta^+(S)}} y_S'.
$$

\noindent Then, we obtain
\begin{align*}
\frac{1}{2} \sum_{e \in F^3} c(e) &= \frac{1}{2} \sum_{e \in F^1} c(e) + \frac{1}{2} \sum_{e \in F^2} c(e) \\
&= \sum_{e \in F^1} \sum_{\substack{S:e \in \delta^+(S)}} y_S + \sum_{e \in F^2} \sum_{\substack{S:e \in \delta^+(S)}} y_S'.
\end{align*}

\noindent Thus, it suffices to show
\begin{equation}
\label{main_eq}
\sum_{e \in F^1} \sum_{\substack{S:e \in \delta^+(S)}} y_S + \sum_{e \in F^2} \sum_{\substack{S:e \in \delta^+(S)}} y_S' \leq \left(\frac{2k-1}{k}\right) \cdot \sum_{\substack{S \subseteq V}} y_S.
\end{equation}

We argue by providing a procedure for covering the edges in $F^3$ in several steps. We first make some definitions. Form the graph of duals $G_{1}=(V_{1},E_{1})$ with $V_{1}$ consisting of all the duals grown in the first and the second phase by the algorithm on $F^1$, and $E_{1}$ consisting of edges of the form $(C^s,C^t)$ if $C^s$ and $C^t$ have grown against each other on an edge in $T_1$. Modify this graph by contracting duals simultaneously grown on an edge into a single dual, and let the growth speed of this dual be the number of contracted duals. Consider a component $T_1$ of $G_1$ and let its vertices be $C_1,\hdots,C_{|T_1|}$ with the corresponding growth speeds $\sigma_1,\hdots,\sigma_{|T_1|}$. Let $\sigma_{max}$ be the largest of these values with the corresponding vertex $C_{max}$. A single step of the procedure covering some portion of the edges of $T_1$ is as follows. For all components $T_1$ of $G_1$, consider increasing all the duals defining $C_{max}$ by an $\epsilon>0$ together with the increments of the neighboring vertices so that they counter the portion of the corresponding edge of $F^1$ in the reverse direction (Recall that we select $\epsilon$ small enough so that the snapshot does not change). Continuing this process in the breadth-first search fashion, consider the increments of the vertices of $T_{1}$ so that all the edges we go over are covered by the same amount in both directions.
\begin{prop}
\label{tree}
$T_{1}$ is a tree.
\end{prop}

\begin{proof}
Take a maximal tree $T$ in $T_{1}$. Let the duals in $T$ be $C_1, \hdots, C_{|T|}$. Assume for a contradiction that there is an edge $e = (C_i,C_j)$ of $T_1$, which is not in $T$. Let $P$ be the path in $F^3$ between the terminals $C_i$ and $C_j$ originate from, which is implied by the path between $C_i$ and $C_j$ in $T$. The existence of $e$ implies that there is another path $P'$ in $F^3$ between the terminals $C_i$ and $C_j$ originate from. In particular, the edge in $F^3$ on which $e$ is defined is distinct from the edges of $P$ by the definition of the graph of duals. This induces a cycle in $F^3$, which is a contradiction.
\end{proof}

Given $C \in V_{1}$, we define $deg_1(C)$ as the graph-theoretic degree of $C$ in $G_{1}$. It is an immediate consequence of Lemma~\ref{dual-one} and Lemma~\ref{dual-one-t} that
\begin{cor}
\label{degree-cor}
For $C \in V_{1}$, $|\Delta(C) \cap F^1| = deg_1(C)$.
\end{cor}

We now make the analogous definitions for $F^2$. Form the graph of duals $G_{2}=(V_{2},E_{2})$ with $V_{2}$ consisting of the duals $y_S'$ defined on $F^2$, and $E_{2}$ consisting of edges of the form $(C^s,C^t)$ if $C^s$ and $C^t$ have grown against each other on an edge in $F^2$. Note that in this graph, the growth speeds of all the vertices are already $1$ by definition. Note also that $V_1$ and $V_2$ might have nonempty intersection. Thus, we extend the procedure described for $G_1$ above to the vertices in $V_2$ by making sure that
\begin{itemize}
\item their increments are compatible with the ones in $V_1$ in the same step,
\item if an increased dual in $V_2$ belongs to the set of duals that we have defined (rather than the ones actually grown by the algorithm), we increase the values of all such duals defined on the current interval together with the duals grown against them.
\end{itemize}

\noindent The second condition is enforced to make sure that there is at least one dual computed by the algorithm corresponding to the duals we have defined in $V_2$, i.e. we can use Proposition~\ref{against-3}.
\begin{prop}
\label{tree-2}
The edges in $V_{2}$ do not share any common vertex.
\end{prop}

\begin{proof}
By the definition of $y_S'$ on $F^2$, there is a distinct dual grown against each dual grown by the algorithm on an edge. Thus, a dual in $y_S'$ cannot be growing against another two duals.
\end{proof}

Given $C \in V_{2}$, we define $deg_2(C)$ as the graph-theoretic degree of $C$ in $G_{2}$. Combining Proposition~\ref{tree-2} with Lemma~\ref{dual-one} and Lemma~\ref{dual-one-t}, we have
\begin{cor}
\label{degree-cor-2}
For $C \in V_{2}$, $|\Delta(C) \cap F^2| = deg_2(C) = 1$.
\end{cor}

After performing a single step of the procedure for all the trees in $G_1$ together with all the duals in $G_2$ affected by their increments, we continue to cover the uncovered portion of the edges in $F^1$ and $F^2$ in the same fashion by recomputing $G_1$ and $G_2$ on the residual graphs. This process terminates since the set of duals is finite and we always select $\epsilon > 0$. We first assume that the uncovered part of $F^1$ is nonempty till the end of the procedure, and argue by induction on the number of steps of the procedure with this assumption. At the beginning of the first step, the values of all the dual values are $0$. Thus, the inequality (\ref{main_eq}) vacuously holds. Assume that it holds at the beginning of some step. Let $T_1$ be a tree of $G_1$ with $|T_1|$ vertices. By Corollary~\ref{degree-cor}, the degree of a vertex in $G_{1}'$ coincides with the number of edges the corresponding dual is incident to in $F^1$. Noting that we have $|T_1|-1$ edges, the amount of increase on the first term of the left hand side of (\ref{main_eq}) for $T_1$ is then
\begin{equation}
\label{1}
\epsilon \sigma_{max} \left(2(|T_1|-1)\right).
\end{equation}

\noindent Since there are $|T_1|$ vertices, the corresponding increase on the right hand side is
\begin{equation}
\label{2}
\epsilon \sigma_{max} \left(\frac{2k-1}{k}\right)|T_1|.
\end{equation}

\noindent On the other hand, we have
\begin{align*}
2(|T_1|-1) = \left(\frac{2(|T_1|-1)}{|T_1|}\right)|T_1|
&\leq \left(\frac{2(2k-1)}{2k}\right)|T_1| \\
&= \left(\frac{2k-1}{k}\right)|T_1|,
\end{align*}

\noindent where the inequality follows from the fact that the number of duals in both phases is at most $k$, i.e. $|T_1| \leq 2k$. Thus, (\ref{1}) is upper bounded by (\ref{2}).

Let $T_2 \subseteq V_2$ be the set of vertices whose values are increased due to the increments in $T_1$. By Corollary~\ref{degree-cor-2}, we have that the degrees of the duals in $T_2$ are all $1$, which is the same as the degrees in $F^2$. In particular, the number of edges in $T_{2}$ is $|T_{2}|/2$. The amount of increase on the second term of the left hand side of (\ref{main_eq}) is then
\begin{equation}
\label{3}
\epsilon \sigma_{max} \left(|T_2|\right).
\end{equation}

\noindent By Proposition~\ref{against-3}, there is at least one dual in $y_S$ for all the $|T_2|/2$ new duals we have defined in $y_S'$. Thus, the corresponding increase on the right hand side is at least
\begin{equation}
\label{4}
\epsilon \sigma_{max} \left(\frac{2k-1}{k}\right)\left(\frac{|T_2|}{2}+1\right).
\end{equation}

\noindent Similar to the previous inequality, we obtain
\begin{align*}
|T_2| = \left(\frac{|T_2|}{\frac{|T_2|}{2}+1}\right) \left(\frac{|T_2|}{2}+1\right)
&\leq \left(\frac{2k}{k+1}\right) \left(\frac{|T_2|}{2}+1\right) \\
&\leq \left(\frac{2k-1}{k}\right) \left(\frac{|T_2|}{2}+1\right),
\end{align*}

\noindent where the first inequality is due to the fact that $|T_2| \leq 2k$, and the second inequality follows since $k \geq 1$. Thus, (\ref{3}) is upper bounded by (\ref{4}). Overall, the inequality (\ref{main_eq}) remains valid at the beginning of the next step of the procedure.

If the uncovered part of $F^1$ becomes empty, for the rest of the procedure, we appropriately select at each step some $T_2 \subseteq V_2$ such that the edges in $F^2$ on which the edges in $T_2$ are defined are simultaneously covered by the algorithm, i.e. $T_2$ is implied by a snapshot of the algorithm. This ensures that $|T_2| \leq 2k$ and Proposition~\ref{against-3} holds for the duals in $T_2$. Mimicking the procedure defined for $F^1$, we increase the values of the duals in $T_2$ by an appropriate $\epsilon > 0$. Then, given a step, the amount of increase on the left hand side of (\ref{main_eq}) is upper bounded by the amount of increase on the right hand side via the exact same argument given above for $T_2$. Thus, the inequality (\ref{main_eq}) holds at the beginning of the next step of the procedure. This completes the induction and hence the proof.
\end{proof}
\section{Tight examples}
We give a tight example essentially putting a lower bound of $2-\frac{1}{k}$ for \texttt{AKR} and \texttt{KLS} in Figure~\ref{tight-AKR} on which the new algorithm finds a near-optimal solution. For this example, the set of duals grown by \texttt{KLS} is the same as that of \texttt{AKR}, and they both select all the edges of cost $1-\epsilon$ before the edges of cost $1/2$ become tight. This makes a total cost of $(2k-1)(1-\epsilon)$. The optimal solution on the other hand consists of all the edges of cost $1/2$ between the pairs of indices from $2$ to $k$ and the edge of cost $1-\epsilon$ between $s_1$ and $t_1$, with a total cost of $k-\epsilon$. In the first phase of the new algorithm, after the duals cover the edges of cost $1/2$ on the $s$-side, the number of duals grown on the edges of cost $1/2$ on the $t$-side becomes $k$. Thus, these edges are also covered within $1/(2k)$ unit of time before all the other edges are covered. Same thing happens in the second phase resulting in the selection of all edges of cost $1/2$, making a total cost of $k$.
\begin{figure}[!t]
\begin{center}
\includegraphics[width=.75\textwidth]{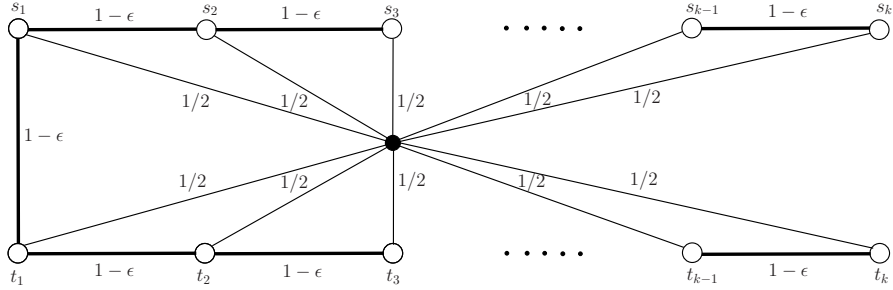}
\end{center}
\caption{A tight example for \texttt{AKR} and \texttt{KLS} on which the new algorithm finds a near-optimal solution}
\label{tight-AKR}
\end{figure}

\begin{figure}[!]
\begin{center}
\includegraphics[width=.75\textwidth]{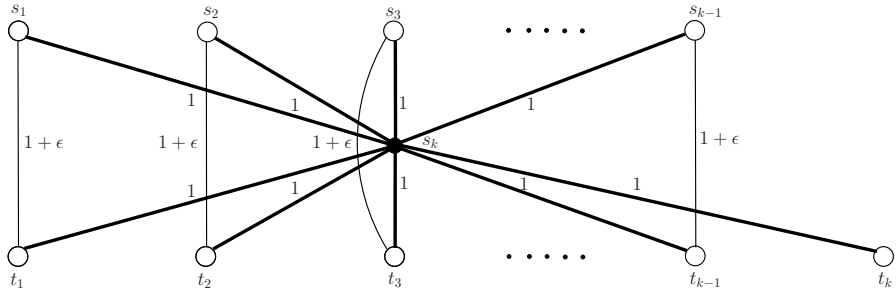}
\end{center}
\caption{A tight example for the new algorithm}
\label{tight}
\end{figure}

A tight example for the new algorithm is given in Figure~\ref{tight}, where the high degree dual is around $s_k$. In the first phase of the algorithm, the set of $k-1$ edges of cost $1$ between the set of terminals $\{s_1,\hdots,s_{k-1}\}$ and $s_k$ are covered in both directions since $s_k$ also grows. Due to this growth, the set of $k$ edges between $s_k$ and the $t$-terminals are also selected in forward direction. In the second phase, these edges are covered in the reverse direction. Thus, the total cost of the solution returned by the algorithm is $2k-1$. The direct edges of cost $1+\epsilon$ between $s_i$ and $t_i$ for $i=1,\hdots,k-1$ remain uncovered throughout the algorithm, which gives an optimal cost of $(k-1)(1+\epsilon)+1$ together with the edge $(s_k,t_k)$.

\section*{Acknowledgment}
We would like to thank David P. Williamson for answering questions on the classical primal-dual algorithm for Steiner forest during the early stages of our investigation. This work was supported by TUBITAK (Scientific and Technological Research Council of Turkey) under Project No. 112E192.

\bibliographystyle{plain}
\bibliography{reference}
\end{document}